%% file: entro2.tex
\documentclass[11pt]{article}
\usepackage{amsmath,amssymb}
\usepackage{amsthm}
\usepackage{mathtools}
\usepackage[noend]{algorithmic}
\usepackage[ruled,vlined]{algorithm2e}
\usepackage{url}
\usepackage{fullpage}
\usepackage{makeidx}
\usepackage{enumerate}
\usepackage[top=1in, bottom=1.25in, left=1in, right=1in]{geometry}
\usepackage{graphicx,float,psfrag,epsfig,caption}
\usepackage[usenames,dvipsnames,svgnames,table]{xcolor}
\definecolor{darkgreen}{rgb}{0.0,0,0.9}
\usepackage[pagebackref,letterpaper=true,colorlinks=true,pdfpagemode=none,citecolor=OliveGreen,linkcolor=BrickRed,urlcolor=BrickRed,pdfstartview=FitH]{hyperref}
\usepackage{epstopdf}
\usepackage{color}
\usepackage{xr}
\usepackage{subfig}
\usepackage{caption}
\usepackage{graphicx}
\usepackage[utf8]{inputenc}

\usepackage{scalerel,stackengine}

\stackMath
\newcommand\reallywidehat[1]{%
\savestack{\tmpbox}{\stretchto{%
  \scaleto{%
    \scalerel*[\widthof{\ensuremath{#1}}]{\kern.1pt\mathchar"0362\kern.1pt}%
    {\rule{0ex}{\textheight}}
  }{\textheight}%
}{2.4ex}}%
\stackon[-6.9pt]{#1}{\tmpbox}%
}
\usepackage[mathscr]{euscript}

\DeclareSymbolFont{rsfs}{U}{rsfs}{m}{n}
\DeclareSymbolFontAlphabet{\mathscrsfs}{rsfs}

\numberwithin{equation}{section}

\newtheoremstyle{myexample} 
    {\topsep}                    
    {\topsep}                    
    {\rm }                   
    {}                           
    {\bf }                   
    {.}                          
    {.5em}                       
    {}  

\newtheoremstyle{myremark} 
    {\topsep}                    
    {\topsep}                    
    {\rm}                        
    {}                           
    {\bf}                        
    {.}                          
    {.5em}                       
    {}  

\newtheorem{theorem}{Theorem}

\theoremstyle{myremark}
\newtheorem{remark}{Remark}[section]

\theoremstyle{myremark}

\theoremstyle{myexample}

\definecolor{darkgreen}{rgb}{0.0, 0.5, 0.0}

\input{commands.tex}

\newcommand{\R}{\mathbb{R}}

\newcommand{\rmd}{\mathrm{d}}

\newcommand{\mF}{\mathcal{F}}

\renewcommand{\hat}{\widehat}

\newcommand{\sMLE}{\mathrm{\tiny{ML}}}
\newcommand{\sB}{\mathrm{\tiny{Bayes}}}
\begin{document}

\title{An Information-Theoretic View of Stochastic Localization}

\author{Ahmed El Alaoui\thanks{ Department of Statistics and Data Science, Cornell University.} \and Andrea Montanari\thanks{Department of Electrical Engineering and
  Department of Statistics, Stanford University.}
}
\date{}

\maketitle

\begin{abstract}
Given a probability measure $\mu$ over $\R^n$, it is often useful to approximate it 
by the convex combination of a small number of probability measures,
such that each component is close to a product measure.
Recently, Ronen Eldan used a stochastic localization argument to prove a general decomposition 
result of this type. In Eldan's theorem, the `number of components' is characterized by the 
entropy of the mixture, and `closeness to product' is characterized by the covariance matrix
of each component.

We present an elementary proof of Eldan's theorem which makes use of an information theory 
(or estimation theory) interpretation. The proof is analogous to the one of an earlier
decomposition result known as the `pinning lemma.'
\end{abstract}

\section{Motivation and result}

Let $\mu$ be an arbitrary Borel probability measure on $\R^n$.  A broadly useful
approach to understanding $\mu$ is to decompose it into simpler components
\begin{align}
\mu =\int_{\Theta} \mu_{\theta}\, \rho(\de\theta) =: \E_{\theta}\mu_{\theta}\, .
\label{eq:FirstDec}
\end{align}
Here $(\Theta,\cF_{\Theta},\rho)$ is an abstract probability space, and, for each
$\theta$, $\mu_{\theta}$ is a probability measure on $\reals^n$.
In this paper, the components $\mu_{\theta}$ will be simple in the sense 
of being close to product measures. 

Of course, the decomposition  \eqref{eq:FirstDec} is always possible:
just take $\Theta=\R^n$, $\rho = \mu$, and $\mu_{\theta}=\delta_{\theta}$ 
(a point mass at $\theta$) for $\theta\in\R^n$. This is a `maximum entropy' decomposition
(all the entropy of $\mu$ is pushed into $\rho$), and is not particularly useful. 
We will be instead interested in decompositions such that $\rho$ has small entropy, and
the $\mu_{\theta}$'s are only approximately product measures.

Low-entropy decompositions have found applications in statistical mechanics \cite{coja2019bethe,austin2019structure},
random graph theory \cite{chatterjee2016nonlinear,eldan2018gaussian},  
random constraint satisfaction problems \cite{ayre2020satisfiability}, 
high-dimensional statistics \cite{coja2018information}, analysis of Markov chains \cite{eldan2021spectral} to name a few areas. A generic construction consists in letting 
$\mu_\theta$ be the conditional distribution of $\bx\sim\mu$ given a small subset of 
the coordinates of $\bx$ (see below for further discussion). 
This leads to the so-called `pinning lemma' which was 
discovered independently in \cite{MontanariSparse,raghavendra2012approximating}. Recently,
 Ronen Eldan \cite{eldan2020taming} pointed out that 
the pining lemma can be suboptimal and proved a general decomposition result with 
better properties. Eldan's proof follows the general approach of `stochastic localization'
which is in turn inspired by ideas in high-dimensional geometry \cite{eldan2016skorokhod}.

The purpose of this note is to present a new interpretation of Eldan's theorem,
leading to an elementary proof. The interpretation has an information-theoretic
(or estimation-theoretic) nature. We consider a noisy communication
channel that takes as input $\bx\sim \mu$ and outputs $\theta$. The distribution
$\mu_{\theta}$ is then the conditional distribution of the channel input given its 
output $\theta$. Note that this is the same interpretation as for the proof of the pinning lemma
in \cite{MontanariSparse}. The main difference is that while in the pinning lemma the channel 
$\bx\to\theta$ is an erasure channel, here we will use a Gaussian channel.

Turning to our construction,  let $\bQ\in \R^{n \times n}$ be a fixed positive
semidefinite matrix,
$\bx \sim \mu$ and define the noisy observation (channel output) $\by$
via
\begin{align}
\by = \sqrt{\tau} \, \bx + \bQ^{1/2}\bz\, ,\label{eq:GaussianChannel}
\end{align} 
where  $\bz \sim \normal(0,\id_n)$, and $\tau$ is uniform in the interval $[1,2]$. 
The random variables $\bx$, $\bz$ and $\tau$ are independent. 
For any fixed $\tau=t$, this is a noisy Gaussian channel, with signal-to-noise
ratio (SNR) $t$. Introducing a random $\tau$ is convenient for the proof, 
but the estimates we prove hold also (with possibly worse constants)
for all $\tau\in[1,2]$, except on a set of small measure. 

We set $\theta = (\by,\tau)$ and $\mu_{\theta}(\cdot) = \mu(\,\cdot\, | \theta)$. 
It is clear that $\mu = \E_{\theta} \mu_{\theta}$ so that the decomposition
\eqref{eq:FirstDec} holds.

Let $\nu$ be a background measure on $\R^n$ such that $\mu$ is absolutely continuous
with respect to $\nu$, i.e. $\mu \ll \nu$.
Recall that the relative entropy of $\mu$ with respect to $\nu$
(Kullback-Leibler divergence) is defined by
\begin{align}
\KL(\mu \|\nu)= \E_{\mu} \log \frac{\rmd \mu}{\rmd \nu}(\bx).
\end{align}
It is easy to see that $\mu_{\theta} \ll \mu \ll \nu$ a.s., so that 
$\KL(\mu_{\theta}\|\nu)$ is well defined.
 
\begin{theorem}[\cite{eldan2020taming}]\label{thm:main}
We have
\begin{align}
\E_{\theta}\cov(\mu_{\theta}) &\preceq \bQ,\label{eq:cov1}\\ 
0\le \E_{\theta} \KL(\mu_{\theta}\|\nu) - \KL(\mu\|\nu) &\le 
\frac{1}{2}\log \det \big(\id_n+2\bQ^{-1}\cov(\mu)\big),\label{eq:entropy}\\
\E_{\theta} \big\{\cov(\mu_{\theta})\bQ^{-1}\cov(\mu_{\theta}) \big\}  
& \preceq \cov(\mu)\label{eq:cov2}.
\end{align} 
\end{theorem}

\begin{remark}
Equations \eqref{eq:cov1} and \eqref{eq:cov2} bound the covariance of the component
measure $\mu_{\theta}$. On the other hand, Eq.~\eqref{eq:entropy}
controls the entropy of the variable $\theta$, which is given by 
the difference between the entropy of the mixture and the average entropy of the components.
More precisely, in information theoretic terms, we control the mutual information
 $\E_{\theta}\KL(\mu_{\theta}\|\nu) - \KL(\mu\|\nu) =\info(\theta;\bx)$ (see proof for definitions).
\end{remark}

\begin{remark}
As mentioned above, the difference between the decomposition presented here
and the pinning lemma of \cite{MontanariSparse} is in the choice of the noisy channel.
Here we use the Gaussian channel \eqref{eq:GaussianChannel} 
(neglecting for a moment the fact that $\tau$ is random). In contrast,
in \cite{MontanariSparse}, the channel is an erasure channel: $y_i=x_i$ with probability $\eps$
and $y_i=\ast$ (an erasure) otherwise, independently across coordinates.
\end{remark}

\begin{remark}
The expectation with respect to $\theta$ in Eqs.~\eqref{eq:cov1} to \eqref{eq:cov2}
is an expectation with respect to $(\tau,\by)$: $\E_{\theta}(\,\cdot\,)=
 \E_{\tau}(\E_{\by|\tau}(\,\cdot\,|\tau))$. The expectation with respect to
 $\tau$ can be eliminated  using Markov inequality. This implies
 that there exists a set $T\subseteq [1,2]$ of Lebesgue measure $|T|\ge 1-3M^{-1}$
 such that, for any $t\in T$, the inequalities 
 \eqref{eq:cov1}, \eqref{eq:entropy}, \eqref{eq:cov2} hold up to a factor $M$
 for $\tau=t$. For instance, the first inequality is replaced by
 $\E(\cov(\mu_{\theta})|\tau=t)\preceq M\bQ$. Note that the conditional expectation
 is simply the expectation with respect to $\by = \sqrt{t}\bx +\bQ^{1/2}\bz$.
 \end{remark}

We finally notice that the above information-theoretic interpretation
 does not only apply to the final decomposition \eqref{eq:FirstDec}, but to the whole
 measure-valued stochastic process defined in \cite{eldan2020taming}. We
 explain this extension in Section \ref{sec:Equivalence}.

\section{Proof}

\subsection{Proof of Eq.~\eqref{eq:cov1}}
 We compare the error of the maximum likelihood estimator of $\bx$ to
  that of the Bayes optimal estimator of $\bx$ given $\theta=(\by,\tau)$: 
  \begin{align*}
\hbx_{\sMLE} = \tau^{-1/2} \, \by~~~~\mbox{and} ~~~~ \hbx_{\sB} = \E[\bx | \by,\tau].
 \end{align*}
Let $\bR \in \R^{n \times n}$ be a fixed PSD matrix, and define
$\|\bv\|_{\bR}^2=\<\bv,\bR\bv\>$. Since $\hbx_{\sB}$ minimizes 
the mean squared error $\E \{\|\bx - \hbx(\theta)\|_{\bR}^2\}$ among all measurable 
estimators  $\hbx$, we have
\begin{align}
\E  \Tr\big[\bR(\bx - \hbx_{\sB})(\bx - \hbx_{\sB})^{\top}\big] \le 
\E  \Tr\big[\bR(\bx - \hbx_{\sMLE})(\bx - \hbx_{\sMLE})^{\top}\big]\, .
\end{align}
The left-hand side is equal to $\E \Tr(\bR\cov(\mu_{\theta}))$, and the right-hand side is equal to
\begin{align}
\E\big[\tau^{-1}\<\bz \bQ^{1/2},\bR \bQ^{1/2} \bz\>\big] = \E[\tau^{-1}] \cdot \Tr(\bR\bQ).
\end{align}
Since the inequality holds for all $\bR \succeq 0$, we conclude that
\begin{align}
\E\, \cov(\mu_{\theta}) \preceq \E[\tau^{-1}] \cdot \bQ \preceq \bQ\, .
\end{align}

\subsection{Proof of~\eqref{eq:entropy}}
This claim follows immediately using some basic inequalities from 
information theory \cite{CoverThomas,dembo1991information}. 

The mutual information of two random variables $X,Y$ is
$\info(X;Y) := \KL(\mu_{X,Y}\|\mu_X\times \mu_Y)$, where we denote by $\mu_{X,Y,\dots}$
the joint law of $X,Y,\dots$. If $\mu_X, \mu_{X|Y}(\, \cdot\, |y)\ll \nu_X$, we have
\begin{align}
\info(X;Y) := \E_{y}\KL(\mu_{X|Y}(\,\cdot\,|y)\|\nu_X)-\KL(\mu_{X}\|\nu_X)\, .
\end{align}
We consider $X=\bx$, $Y=\by$, under their joint distribution given $\tau=t$. 
In other words $\bx\sim \mu$ and $\by$ is given by Eq.~\eqref{eq:GaussianChannel} 
with $\tau=t$. We have $\mu_{\theta}=\mu_{X|Y}(\,\cdot\,|\by)$.
Using non-negativity of the mutual information, we have
\begin{align}
0\le \info(\bx;\by) = \E_{\by}\KL(\mu_{\theta}\|\nu)-\KL(\mu\|\nu)\, .
\end{align}
which yields the first inequality in Eq.~\eqref{eq:entropy}. Inverting the role of $X,Y$
in the definition of mutual information, and letting $\overline{\nu}$ be, for instance, 
 the standard Gaussian measure we get 
\begin{align}
\info(\bx;\by) = \E_{\bx}\KL(\mu_{\by|\bx}(\,\cdot\,|\bx)\|\overline{\nu})-
\KL(\mu_{\by}\|\overline{\nu})\, .
\end{align}
Recall the definition of differential entropy of a random variable $X$ with density $f$ with respect to the Lebesgue
measure: $h(X) :=-\int f(x)\log f(x) \de x$, and $h(X|Y) = h(X,Y)-h(Y)$. 
We then have from the last display
\begin{align}
\info(\bx;\by) &=  h(\by) - h(\by|\bx)\, . 
\end{align}
The differential entropy of an $n$-dimensional Gaussian vector $\bg$ with covariance $\bSigma$
is $h(\bg) = (n/2)\log(2\pi e)+(1/2)\Tr\log \bSigma$. Therefore,  
\[h(\by|\bx) =  \frac{n}{2}\log(2\pi e)+ \frac{1}{2}\Tr\log \bQ \, .\]
Moreover, the Gaussian distribution maximizes the differential entropy among all those distributions with the same covariance, hence
\begin{align} 
h(\by) \le  \frac{n}{2}\log(2\pi e)+ \frac{1}{2}\Tr\log \cov(\mu_{\by}) \, . 
\end{align}
Since $\cov(\mu_{\by}) = t \cov(\mu) + \bQ$, we obtain
\begin{align}
\info(\bx;\by) &\le \frac{1}{2}\Tr\log \cov(\mu_{\by})-  \frac{1}{2}\Tr\log \bQ \\
&=   \frac{1}{2}\log \det \big( \id_n + t \bQ^{-1} \cov(\mu)\big) \, .
\end{align}
The claim then follows since $t \le 2$.

\subsection{Proof of~\eqref{eq:cov2}}

  Fixing $\tau=t$ a non-random value, we write $\by_t$ for the corresponding output of 
  channel \eqref{eq:GaussianChannel}, namely $\by_t=\sqrt{t}\bx+\bQ^{1/2}\bz$. 
  We denote by $\mu_{\by_t}(\de\bx)$ for the conditional distribution of 
  $\bx$ given $\by_t$. In order to emphasize its
  dependence on $t$, we will write   $\muc :=\mu_{\by_t=\sqrt{t}\bx_0+\bQ^{1/2}\bz}$.
  Simplifying Bayes formula, we get 
  \begin{align}
  \muc(\de\bx) = \frac{1}{Z_{\bx_0,\bz,t}}\, 
  \exp\Big\{-\frac{t}{2}\|\bx-\bx_0\|_{\bQ^{-1}}^2 +\sqrt{t} \<\bz,\bQ^{-1/2}\bx\>\Big\}\,
  \mu(\de\bx) \, .
  \end{align} 
  Here we use the notation $\|\bv\|^2_{\bA}:= \<\bv,\bA\bv\>$.  
  
  Throughout this proof, given a measure $\nu(\de\bx)$ and function $\psi_1(\bx),\psi_2(\bx)$, 
  we use the  shorthands
  $\nu(\psi(\bx)) := \int\psi(\bx)\nu(\de\bx)$ and $\nu(\psi_1(\bx);\psi_2(\bx)):=
  \nu(\psi_1(\bx)\cdot\psi_2(\bx))-\nu(\psi_1(\bx))\nu(\psi_2(\bx))$.
  We then have:
  \begin{align}
  \frac{\de\phantom{t}}{\de t}\muc(\bx) = -\frac{1}{2}\muc\big(\bx;\|\bx-\bx_0\|_{\bQ^{-1}}^2\big)
  +\frac{1}{2\sqrt{t}}\muc\big(\bx;\<\bz,\bQ^{-1/2}\bx\>\big)\, .
  \end{align}
  (Note that $\muc(\bx)$ is the mean of the vector $\bx$.)
  
  Given a matrix $\bR\succeq 0$, we define the minimum mean square error
  \begin{align}
  \mmse(t) &:= \E\Big\{\big\|\bx-\E(\bx|\by_t)\big\|_{\bR}^2\Big\}\\
  & = \Tr\big(\bR\,\cov(\mu)\big)- \E_{\bx_0,\bz}\<\muc(\bx),\bR\muc(\bx)\>\, .
  \end{align}
  Differentiating, and using the formula above to differentiate $\muc(\bx)$,
  we get 
  \begin{align}
  \dert\mmse(t) & = -2 \cdot\E_{\bx_0,\bz}\<\muc(\bx),\bR\dert \muc(\bx)\> \\
  &=: A-B_1+B_2\, ,\label{eq:dtmmse}\\
  A&:= \E_{\bx_0,\bz} \big\{\<\muc(\bx),\bR\muc(\bx;\|\bx-\bx_0\|^2_{\bQ^{-1}})\>\big\}\, ,
  \\
  B_1 &:= \frac{1}{\sqrt{t}}\E_{\bx_0,\bz} \big\{\<\muc(\bx),\bR\muc(\bx\<\bz,\bQ^{-1/2}\bx\>)\>\big\}\, ,\\
  B_2 &:= \frac{1}{\sqrt{t}}\E_{\bx_0,\bz} \big\{\<\muc(\bx),\bR\muc(\bx)\>\muc(\<\bz,\bQ^{-1/2}\bx\>)\big\}\, .
  \end{align}
  We next use integration by parts (Stein's Lemma) to simplify the terms $B_1,B_2$.
  It is useful to note that
  \begin{align}
  \frac{1}{\sqrt{t}}\nabla_{\bz}\muc\big(\psi(\bx)\big) = \muc\big(\psi(\bx);\bQ^{-1/2}\bx\big)\, .
  \end{align}
  Writing for simplicity $\mus:=\muc$ and $\E=\E_{\bx_0,\bz}$, we get
   \begin{align}
  B_1=&\E\big\{\mus(\bx)^\top\bR\mus(\bx\|\bx\|^2_{\bQ^{-1}})\big\}  - \E \big\{\mus(\bx)^\top
  \bR \mus(\bx\bx^\top)\bQ^{-1}\mus(\bx)\big\}\label{eq:B1EQ}\\
  &+\E\Tr\big\{\bQ^{-1}\big[\mus(\bx\bx^\top)-\mus(\bx)\mus(\bx^\top)\big]
  \bR\mus(\bx\bx^\top)\big\}\, ,\nonumber\\
  B_2= & \E\big\{\mus(\bx)^\top\bR\mus(\bx)\mus(\|\bx\|_{\bQ^{-1}}^2)\big\}-
  \E\big\{\mus(\bx)^\top \bR\mus(\bx)\mus(\bx)^\top \bQ^{-1}\mus(\bx)\big\}\label{eq:B2EQ}\\
  &+2\E\big\{\mus(\bx)^{\sT}\bQ^{-1}[\mus(\bx\bx^\top)-\mus(\bx)\mus(\bx)^{\top}]\bR\mus(\bx)\big\}\, .
  \nonumber
   \end{align}
  Finally, by the tower property of conditional expectation
  \begin{align}
  A& = \E\big\{\mus(\bx)^\top\bR\mus(\bx;\|\bx\|^2_{\bQ^{-1}})\>\big\}-
  2\E \big\{\mus(\bx)^\top\bR\mus(\bx;\bx^{\top})\bQ^{-1}\bx_0\big\}\\
  & = \E\big\{\mus(\bx)^\top\bR\mus(\bx;\|\bx\|^2_{\bQ^{-1}})\>\big\}-
  2\E \big\{\mus(\bx)^\top\bR[\mus(\bx\bx^{\top})-\mus(\bx)\mus(\bx)^\top]\bQ^{-1}\mus(\bx)\big\}\, .
  \label{eq:AEQ}
  \end{align}
  We next substitute Eqs.~\eqref{eq:B1EQ}, \eqref{eq:B2EQ}, and \eqref{eq:AEQ}
 in Eq.~\eqref{eq:dtmmse} to get
  \begin{align}
  \dert\mmse(t) = -\E\Tr\big\{\cov(\mus)\bQ^{-1}\cov(\mus)\bR\big\} \, .
  \end{align}
  We next integrate this identity for $t\in[1,2]$, to get
  \begin{align*}
  \int_1^2\E\Tr\big\{\cov(\mus)\bQ^{-1}\cov(\mus)\bR\big\} \de t &=
  \mmse(1)-\mmse(2)\le \mmse(0)\\
  & =  \Tr( \cov(\mu)\bR)\,.
  \end{align*}
  Finally, the integral over $t$ can be interpreted as an expectation over 
  $\tau\sim\Unif([1,2])$, whence
  \begin{align*}
\E\Tr\big\{\cov(\mu_{\theta})\bQ^{-1}\cov(\mu_{\theta})\bR\big\} \le
\Tr( \cov(\mu)\bR)\, .
\end{align*}
Since this holds for any $\bR\succeq 0$, we proved Eq.~\eqref{eq:cov2}.

\section{Extension to the stochastic localization process}
\label{sec:Equivalence}

Another way of writing the output of the Gaussian channel is as follows: Let $\bx \sim \mu$ and $(\bB_t)_{t \ge 0}$ be a standard Brownian motion in $\R^n$. Further, let 
\begin{equation}
\bar{\by}_t = t \bx + \bQ^{1/2} \bB_t.
\end{equation}
Then it is clear that $\bar{\by}_t$ has the same law as the vector $\sqrt{t}\by$, so the conditional distribution of $\bx$ given $\by$ has the same law as the probability measure 
\begin{equation}\label{eq:mut}
\mu_t(\de \bx) = \frac{1}{Z_t} 
\exp\Big\{\langle \bar{\by}_t ,  \bx\rangle_{\bQ^{-1}} - \frac{t}{2} \|\bx\|_{\bQ^{-1}}^2\Big\}  \mu(\de \bx) \, .
\end{equation}

We will show that the measure-valued process $(\mu_t)_{t \ge 0}$ satisfies the stochastic 
localization SDE of Eldan~\cite{eldan2020taming}, as stated below.
\begin{theorem} 
Write $L_t$ for the likelihood ratio process of $\mu_t$ with respect to $\mu$: 
\begin{equation}
L_t(\bx) := \frac{\de \mu_t}{\de \mu}(\bx) \, .
\end{equation}
Then there exist a Brownian motion $(\bW_t)_{t \ge 0}$ adapted to the filtration generated by $(\bar{\by}_t)_{t \ge 0}$, such that for all $\bx \in \R^n$ and $t \ge 0$, we have
\begin{equation}\label{eq:SL_SDE}
\de L_t (\bx) = L_t(\bx) \big\langle \bx - \ba_t , \bQ^{-1/2} \de \bW_t \big\rangle \, ,~~~~\mbox{and}~~~L_0(\bx)=1 \, ,
\end{equation}
where
\begin{equation}
\ba_t = \E\big\{\bx | \bar{\by}_t\big\} = \int \bx \,  \mu_t(\de\bx) \, .
\end{equation}
\end{theorem}
\begin{proof}  
We use the representation~\eqref{eq:mut} to write
\begin{align}\label{eq:log_L}
\de \log L_t(\bx) = \langle \de \bar{\by}_t ,  \bx\rangle_{\bQ^{-1}} - \frac{1}{2} \|\bx\|_{\bQ^{-1}}^2 \de t - \de \log Z_t\, .
\end{align}
Let us write $h_t(\bx)$ for the expression appearing in the exponent in Eq.~\eqref{eq:mut}. By It\^{o}'s formula we have
\begin{align}
\de Z_t &= \int_{\R^n} \Big(\< \bx , \bQ^{-1} \de \bar{\by}_t \>- \frac{1}{2}\|\bx\|_{\bQ^{-1}}^2 
\de t \Big)  e^{h_t(\bx)} \mu(\de\bx)
 + \frac{1}{2}  \Big(\int_{\R^n} \|\bx\|_{\bQ^{-1}}^2  e^{h_t(\bx)} \mu(\de\bx) \Big) \de t \\
&=  \Big\langle \int_{\R^n} \bx e^{h_t(\bx)} \mu(\de \bx) , \bQ^{-1} \de \bar{\by}_t \Big\rangle \, . 
\end{align}
With the notation $ \ba_t =\E\big\{\bx | \bar{\by}_t\big\} = \int \bx  \, \mu_t(\de\bx)$ 
(and writing $[Z]_t$ for the quadratic variation process), we obtain
\begin{align}
\de \log Z_t &= \frac{\de Z_t}{Z_t} - \frac{1}{2} \frac{\de [Z]_t}{Z_t^2} \\
&= \big\langle \ba_t ,  \de \bar{\by}_t \big\rangle_{\bQ^{-1}} - \frac{1}{2} \| \ba_t \|_{\bQ^{-1}}^2 \de t \, .
\end{align}
Substituting this in Eq.~\eqref{eq:log_L} we have
\begin{align}
\de \log L_t(\bx) &= \big\langle \bx - \ba_t ,  \de \bar{\by}_t \big\rangle_{\bQ^{-1}} - \frac{1}{2} \Big(\|\bx\|_{\bQ^{-1}}^2 - \| \ba_t \|_{\bQ^{-1}}^2\Big)\de t \\
&= \big\langle \bx - \ba_t ,  \de \bar{\by}_t -\ba_t \de t \big\rangle_{\bQ^{-1}} - \frac{1}{2} \|\bx - \ba_t\|_{\bQ^{-1}}^2 \de t \, .
\end{align}
It remains to understand the law of the process $\bar{\by}_t - \int_0^t \ba_s \de s$. 
We let $\mF_t = \sigma(\{\bar{\by}_s: 0 \le s \le t\})$ and $\mathbb{F} = (\mF_t)_{ t \ge 0}$. 
It is known  that the process $(\bW_t)_{t \ge 0}$ defined by
\begin{equation}
\bW_t := \bQ^{-1/2}\Big(\bar{\by}_t - \int_0^t \E\big\{ \bx | \mF_s \big\} \de s \Big)\, ,
\end{equation}
is an $\mathbb{F}$-adapted Brownian motion. 
See for instance Theorem 7.12 in~\cite{liptser1977statistics}, 
or Theorem 5.13 in~\cite{le2016brownian}. 
We conclude by noticing that $\E\big\{ \bx | \mF_t \big\} = \ba_t$ since $\bar{\by}_t$ is a sufficient statistic for $\bx$ under $\mu_t$; c.f.\ Eq.~\eqref{eq:mut}, therefore    
\begin{equation}
\de \log L_t(\bx) = \big\langle \bx - \ba_t ,  \bQ^{-1/2}\de \bW_t \big\rangle - \frac{1}{2} \|\bx - \ba_t\|_{\bQ^{-1}}^2 \de t \, , 
 \end{equation}
and we obtain~\eqref{eq:SL_SDE} by applying It\^{o}'s formula once more. 
\end{proof}

\section*{Acknowledgements}

We thank Ronen Eldan for a stimulating discussion. This work was carried out 
while the authors were (virtually) participating in the Fall 2020 program on Probability, Geometry and Computation in High Dimensions at the Simons Institute for the Theory of Computing.
A.M. was partially supported by  the NSF grant CCF-2006489 and the ONR grant 
N00014-18-1-2729.

\newpage

\bibliographystyle{amsalpha}
\providecommand{\bysame}{\leavevmode\hbox to3em{\hrulefill}\thinspace}
\providecommand{\MR}{\relax\ifhmode\unskip\space\fi MR }
\providecommand{\MRhref}[2]{%
  \href{http://www.ams.org/mathscinet-getitem?mr=#1}{#2}
}
\providecommand{\href}[2]{#2}

\end{document}

%% file: commands.tex
\newcommand{\bea}{\begin{eqnarray}}
\newcommand{\eea}{\end{eqnarray}}
\newcommand{\<}{\langle}
\renewcommand{\>}{\rangle}
\newcommand{\E}{{\mathbb E}}

\def\Unif{{\sf Unif}}

\def\eps{{\varepsilon}}

\def\id{{\boldsymbol{I}}}

\def\muc{\mu_{\bx_0,\bz,t}}
\def\mus{\mu_{t}}

\def\bSigma{{\boldsymbol{\Sigma}}}

\def\bQ{{\boldsymbol{Q}}}

\def\dert{\frac{\de\phantom{t}}{\de t}}

\def\info{{\rm I}}
\def\KL{{\rm D}}

\def\bg{{\boldsymbol{g}}}

\def\cF{{\mathcal F}}

\def\mmse{{\rm mmse}}

\def\reals{{\mathbb R}}

\def\normal{{\sf N}}

\def\sT{{\sf T}}

\def\bv{{\boldsymbol{v}}}
\def\bz{{\boldsymbol{z}}}
\def\bx{{\boldsymbol{x}}}
\def\ba{{\boldsymbol{a}}}

\def\bA{\boldsymbol{A}}
\def\bB{\boldsymbol{B}}

\def\hbx{\boldsymbol{\hat{x}}}

\def\de{{\rm d}}

\def\bW{\boldsymbol{W}}

\def\E{{\mathbb E}}

\def\<{\langle}
\def\>{\rangle}
\def\Tr{{\rm Tr}}

\def\by{{\boldsymbol{y}}}

\def\cov{{\rm Cov}}

\def\b0{{\boldsymbol{0}}}

\def\bR{{\boldsymbol R}}

%% file: entro2.bbl
\begin{thebibliography}{ACOGM20}

\bibitem[ACOGM20]{ayre2020satisfiability}
Peter Ayre, Amin Coja-Oghlan, Pu~Gao, and No{\"e}la M{\"u}ller, \emph{The
  satisfiability threshold for random linear equations}, Combinatorica
  \textbf{40} (2020), no.~2, 179--235.

\bibitem[Aus19]{austin2019structure}
Tim Austin, \emph{The structure of low-complexity gibbs measures on product
  spaces}, The Annals of Probability \textbf{47} (2019), no.~6, 4002--4023.

\bibitem[CD16]{chatterjee2016nonlinear}
Sourav Chatterjee and Amir Dembo, \emph{Nonlinear large deviations}, Advances
  in Mathematics \textbf{299} (2016), 396--450.

\bibitem[COKPZ18]{coja2018information}
Amin Coja-Oghlan, Florent Krzakala, Will Perkins, and Lenka Zdeborov{\'a},
  \emph{Information-theoretic thresholds from the cavity method}, Advances in
  Mathematics \textbf{333} (2018), 694--795.

\bibitem[COP19]{coja2019bethe}
Amin Coja-Oghlan and Will Perkins, \emph{Bethe states of random factor graphs},
  Communications in Mathematical Physics \textbf{366} (2019), no.~1, 173--201.

\bibitem[CT91]{CoverThomas}
Thomas.~M. Cover and Joy~A. Thomas, \emph{{Elements of Information Theory}},
  Wiley, 1991.

\bibitem[DCT91]{dembo1991information}
Amir Dembo, Thomas~M Cover, and Joy~A Thomas, \emph{Information theoretic
  inequalities}, IEEE Transactions on Information theory \textbf{37} (1991),
  no.~6, 1501--1518.

\bibitem[EKZ21]{eldan2021spectral}
Ronen Eldan, Frederic Koehler, and Ofer Zeitouni, \emph{A spectral condition
  for spectral gap: fast mixing in high-temperature ising models}, Probability
  Theory and Related Fields (2021), 1--17.

\bibitem[Eld16]{eldan2016skorokhod}
Ronen Eldan, \emph{Skorokhod embeddings via stochastic flows on the space of
  gaussian measures}, Annales de l'Institut Henri Poincar{\'e},
  Probabilit{\'e}s et Statistiques, vol.~52, Institut Henri Poincar{\'e}, 2016,
  pp.~1259--1280.

\bibitem[Eld18]{eldan2018gaussian}
\bysame, \emph{Gaussian-width gradient complexity, reverse log-sobolev
  inequalities and nonlinear large deviations}, Geometric and Functional
  Analysis \textbf{28} (2018), no.~6, 1548--1596.

\bibitem[Eld20]{eldan2020taming}
\bysame, \emph{Taming correlations through entropy-efficient measure
  decompositions with applications to mean-field approximation}, Probability
  Theory and Related Fields \textbf{176} (2020), no.~3, 737--755.

\bibitem[LG16]{le2016brownian}
Jean-Fran{\c{c}}ois Le~Gall, \emph{Brownian motion, martingales, and stochastic
 calculus}, Springer, 2016.

\bibitem[LS77]{liptser1977statistics}
Robert~Shevilevich Liptser and Al'bert~Nikolaevich Shiriaev, \emph{Statistics
  of random processes: General theory}, vol. 394, Springer, 1977.
  
\bibitem[Mon08]{MontanariSparse}
Andrea Montanari, \emph{{Estimating random variables from random sparse
  observations}}, European Transactions on Telecommunications \textbf{19}
  (2008), 385--403.

\bibitem[RT12]{raghavendra2012approximating}
Prasad Raghavendra and Ning Tan, \emph{Approximating csps with global
  cardinality constraints using sdp hierarchies}, Proceedings of the
  twenty-third annual ACM-SIAM symposium on Discrete Algorithms, SIAM, 2012,
  pp.~373--387.

\end{thebibliography}
